\documentclass[10pt]{article}

\usepackage[utf8]{inputenc}

\usepackage{epsf}
\usepackage{amsmath}

\allowdisplaybreaks

\usepackage[showframe=false]{geometry}
\usepackage{changepage}

\usepackage{epsfig}
\usepackage{amssymb}

\usepackage{amsthm}
\usepackage{setspace}
\usepackage{cite}
\usepackage{mcite}

\usepackage{algorithmic}  
\usepackage{algorithm}

\usepackage{shadow}
\usepackage{fancybox}
\usepackage{fancyhdr}

\usepackage{color}
\usepackage[usenames,dvipsnames,svgnames,table]{xcolor}

\definecolor{mypurple}{rgb}{.4,.0,.5}

\usepackage{xcolor}

\usepackage{color}

\definecolor{darkgreen}{rgb}{0, 0.4,0}

\definecolor{purplebrown}{rgb}{0.5,0.1,0.6}

\newcommand{\bl}[1]{\textcolor{blue}{#1}}

\definecolor{shadebrown}{rgb}{0.1,0.1,0.9}
\definecolor{lightblue}{rgb}{0.2,0,1}

\usepackage{fancybox}
\usepackage{graphicx}
\usepackage{epstopdf}
\usepackage{epsfig}
\usepackage{wrapfig}
\usepackage{subfigure}

\usepackage{xcolor}

\usepackage{tcolorbox}
\tcbuselibrary{skins}

\newtcbox{\xmybox}{on line,
arc=7pt,
before upper={\rule[-3pt]{0pt}{10pt}},boxrule=0pt,
boxsep=0pt,left=6pt,right=6pt,top=0pt,bottom=0pt,enhanced, coltext=blue, colback=white!10!yellow}

\newtcbox{\xmyboxa}{on line,
arc=7pt,
before upper={\rule[-3pt]{0pt}{10pt}},boxrule=0pt,
boxsep=0pt,left=6pt,right=6pt,top=0pt,bottom=0pt,enhanced, colback=white!10!yellow}

\newtcbox{\xmyboxb}{on line,
arc=7pt,
before upper={\rule[-3pt]{0pt}{10pt}},boxrule=1pt,colframe=darkgreen!100!blue,
boxsep=0pt,left=6pt,right=6pt,top=0pt,bottom=0pt,enhanced, colback=white!10!yellow}

\newtcbox{\xmyboxc}{on line,
arc=7pt,
before upper={\rule[-3pt]{0pt}{10pt}},boxrule=.7pt,colframe=blue!100!blue,
boxsep=0pt,left=6pt,right=6pt,top=0pt,bottom=0pt,enhanced, coltext=blue, colback=white!10!yellow}

\newtcbox{\xmytboxa}{on line,
arc=7pt,
before upper={\rule[-3pt]{0pt}{10pt}},boxrule=.0pt,colframe=pink!50!yellow,
boxsep=0pt,left=6pt,right=6pt,top=0pt,bottom=0pt,enhanced, coltext=white, colback=blue!40!red}

\newtcbox{\xmytboxb}{on line,
arc=7pt,
before upper={\rule[-3pt]{0pt}{10pt}},boxrule=.0pt,colframe=pink!50!yellow,
boxsep=0pt,left=6pt,right=6pt,top=0pt,bottom=0pt,enhanced, coltext=white, colback=white!40!green}

\usepackage[hyphens]{url}

\usepackage[colorlinks=true,
            linkcolor=black,
            urlcolor=blue,
            citecolor=purple]{hyperref}

\usepackage{breakurl}

\def\y{{\bf y}}
\def\v{{\bf v}}
\def\x{{\bf x}}

\def\x{{\mathbf x}}

\def\v{{\bf v}}
\def\x{{\bf x}}
\def\y{{\bf y}}

\def\be{\begin{equation}}
\def\ee{\end{equation}}
\def\ba{\left[\begin{array}}
\def\ea{\end{array}\right]}

\def\v{{\bf v}}
\def\x{{\bf x}}
\def\y{{\bf y}}

\def\1{{\bf 1}}

\def\0{{\bf 0}}

\def\erf{\mbox{erf}}
\def\erfc{\mbox{erfc}}

\def\mR{{\mathbb R}}

\newtheorem{theorem}{Theorem}

\begin{document}


\begin{singlespace}

\title {Algorithmic random duality theory -- large scale CLuP 
}
\author{
\textsc{Mihailo Stojnic
\footnote{e-mail: {\tt flatoyer@gmail.com}} }}
\date{}
\maketitle

\centerline{{\bf Abstract}} \vspace*{0.1in}

Based on our \bl{\textbf{Random Duality Theory (RDT)}}, in a sequence of our recent papers \cite{Stojnicclupint19,Stojnicclupcmpl19,Stojnicclupplt19}, we introduced a powerful algorithmic mechanism (called \bl{\textbf{CLuP}}) that can be utilized to solve \textbf{\emph{exactly}} NP hard optimization problems in polynomial time. Here we move things further and utilize another of remarkable RDT features that we established in a long line of work in \cite{StojnicCSetam09,StojnicCSetamBlock09,StojnicISIT2010binary,StojnicDiscPercp13,StojnicUpper10,StojnicGenLasso10,StojnicGenSocp10,StojnicPrDepSocp10,StojnicRegRndDlt10,Stojnicbinary16fin,Stojnicbinary16asym}. Namely, besides being stunningly precise in characterizing the performance of various random structures and optimization problems, RDT simultaneously also provided an almost unparallel way for creating computationally efficient optimization algorithms that achieve such performance. One of the keys to our success was our ability to transform the initial \textbf{\emph{constrained}} optimization into an \textbf{\emph{unconstrained}} one and in doing so greatly simplify things both conceptually and computationally. That ultimately enabled us to solve a large set of classical optimization problems on a very large scale level. Here, we demonstrate how such a thinking can be applied to CLuP as well and eventually utilized to solve pretty much any problem that the basic CLuP from \cite{Stojnicclupint19,Stojnicclupcmpl19,Stojnicclupplt19} can solve. Since this is the first paper in this direction we focus on the utilization of the \textbf{\emph{large scale}} CLuP for solving the famous MIMO ML detection problem and show that it can easily handle problems of dimensions of \emph{\textbf{several thousands}} with theoretically minimal -- \textbf{\emph{quadratic}} per iteration -- complexity that includes only a \bl{\emph{\textbf{single matrix-vector multiplication}}}. A solid set of numerical experiments demonstrates a rather remarkable agreement between the simulated results and the theoretical predictions.

As mentioned on various occasions in \cite{Stojnicclupint19,Stojnicclupcmpl19,Stojnicclupplt19}, CLuP mechanisms are very generic and can be used to attack a tone of problems in various other scientific fields. In some of our companion papers we will discuss how similar ideas can be adapted to all of these scenarios as well.

\vspace*{0.25in} \noindent {\bf Index Terms: Large scale CLuP; ML - detection; MIMO systems; Algorithms; Random duality theory}.

\end{singlespace}

\section{Introduction}
\label{sec:back}

In \cite{Stojnicclupint19,Stojnicclupcmpl19,Stojnicclupplt19}, we recently introduced a powerful concept called \bl{\textbf{CLuP}} as a mechanism of our \bl{\textbf{Random Duality Theory (RDT)}} to attack hard optimization problems. A large set of numerical and theoretical results presented in \cite{Stojnicclupint19,Stojnicclupcmpl19,Stojnicclupplt19} demonstrated a rather fascinating phenomenon. Namely, achieving the so-called ML performance at the output detection of the MIMO systems (the well known NP hard problem within the classical complexity theory) is actually possible in polynomial time. Since the concepts presented in \cite{Stojnicclupint19,Stojnicclupcmpl19,Stojnicclupplt19} are very general, we could have chosen a tone of other problems as introductory examples where we would demonstrate the power of CLuP and RDT. We chose the so called MIMO ML detection for out initial presentation due to its enormous popularity in a host of scientific fields ranging from signal processing and information theory to statistics and machine learning and many others. To parallel our introductory presentations from \cite{Stojnicclupint19,Stojnicclupcmpl19,Stojnicclupplt19} and to facilitate the presentation of the main ideas that we introduce later on here we will again focus on the very same MIMO ML problem. Along the same lines, we first briefly recall the definition of the problem. However, due to its popularity and the fact the we have already discussed this very same problem on quite a few occasions in \cite{Stojnicclupint19,Stojnicclupcmpl19,Stojnicclupplt19} (and earlier in e.g. \cite{StojnicBBSD05,StojnicBBSD08}) we will, as usual, try to avoid repeating many of the details presented in these papers and instead focus on the key differences.

One typically models the noise corrupted linear MIMO system in the following way:
\begin{eqnarray}\label{eq:linsys1}
\y=A\x_{sol}+\sigma\v.
\end{eqnarray}
It is rather obvious, but for the completeness, $\x_{sol}\in\mR^n$ is the vector at the input of the MIMO system, $A\in\mR^{m\times n}$ is the matrix that models the connections (to be precise, exactly $mn$ of them) between the multiple input and the multiple output of the system, and $\v\in\mR^m$ is the so-called noise vector that is of additive type and is added on top of what is obtained at the output of the system (the noise is typically scaled by a factor $\sigma$ that controls the so-called signal-to-noise (SNR) ratio). Of course, it goes almost without saying that the vector $\y\in\mR^m$ is the vector one finally obtains at the output of the system after the noise is added. The linear MIMO models defined above are of course well known and belong to the group of some of the most fundamental mathematical models in many scientific fields, information theory, signal processing, linear estimation, and statistics, just to name a few.

It is not that hard to guess that one of the key problems in the above linear MIMO models is the recovery of the input vector $\x_{sol}$ at the output of the system. Such a recovery is typically called MIMO detection. Depending on the scenario where the models are used, the recovery/detection of $\x_{sol}$ might be under the premise that the systems matrix $A$ is: 1) unknown (non-coherent detection) and 2) known (coherent detection). In this paper we will focus on the coherent type of detection assuming that the matrix $A$ is known at the system output. The structure of the system matrix $A$, the noise vector $\v$, and the input vector/signal $\x_{sol}$ also play important role in one's ability to ultimately recover $\x_{sol}$ at the output (or as it is often called, the receiving end). Among the key structural features are the dimensions and the type of the entries of these vectors. Here we will assume the so-called linear regime, i.e. that $m=\alpha n$ where $\alpha>0$ is a real number and that both, $n$ and $m$, are large (while we ill be interested in all scenarios for any $\alpha>0$, the experienced reader will immediately recognize that the most interesting and hardest to handle are those where $\alpha<1$; precisely those will be of our predominant interest as well). We will also consider the so-called typical Gaussian statistical scenario where the elements of matrix $A$ and vector $\v$ are i.i.d. standard normal (i.e. Gaussian) random variables (we, of course, add right here at the beginning that Gaussianity of $A$ is chosen for the simplicity of the presentation; however, it can be replaced with pretty much any distribution with very mild moment conditions and all our results will continue to be in place). Within the framework that we will consider below, the structure of $\x_{sol}$ can be pretty much anything. In this paper though, we will typically focus on the so-called binary scenario (typically seen in digital communications and multi-antenna MIMO systems). However in some of our companion papers we will discuss various other $\x_{sol}$ structures.

Under the above assumptions, one of the most famous statistical/information theoretic ways to recover/estimate $\x_{sol}$ is through the so-called ML criterion which amounts to solving the following optimization problem
\begin{eqnarray}\label{eq:ml1}
\hat{\x}=\min_{\x\in{\cal X}}\|\y-A\x\|_2.
\end{eqnarray}
It is rather clear that ${\cal X}$ stands for the set of possible $\x_{sol}$ at the input of the system. As mentioned above, to parallel the expositions from \cite{Stojnicclupint19,Stojnicclupcmpl19,Stojnicclupplt19}, we will here consider the standard binary scenario which assumes ${\cal X}=\{-\frac{1}{\sqrt{n}},\frac{1}{\sqrt{n}}\}^n$. However, as was the case for the mechanisms presented in \cite{Stojnicclupint19,Stojnicclupcmpl19,Stojnicclupplt19}, the mechanisms that we present below are very generic and hold for pretty much any type of set ${\cal X}$ (also as mentioned above, in some of our companion papers we will discuss what form the results that we present below for the binary ${\cal X}$ take when ${\cal X}$ has various different structures).

Now that we established the optimization in (\ref{eq:ml1}) as the key problem of our interest in this paper, one wonders how it can be solved. As stated on multiple occasions in \cite{Stojnicclupint19,Stojnicclupcmpl19,Stojnicclupplt19} the assumed binary structure of ${\cal X}$ makes the optimization problem in (\ref{eq:ml1}) among the hardest well known optimization problems (in particular so when $\alpha<1$). A lot of work over last several decades has been done in an attempt to solve this problem either \emph{approximately} through various heuristics or \emph{exactly} through precise optimization algorithms designed to achieve the global optimum in (\ref{eq:ml1}).

As our goal in this paper is a discussion related to a particular type of algorithms we will leave a more thorough discussion of all relevant prior work regarding the known techniques that can be used for solving (\ref{eq:ml1}) to survey papers. Here we will just briefly mention a few papers that are most closely related to what we discuss below. Namely, among the most popular continuous relaxation type of heuristics are those that are well known in the optimization theory communities and that we discussed as parts of strategies that we designed in \cite{StojnicBBSD08,StojnicBBSD05}. These techniques relate to the so-called Ball, Polytope, and SDP convex relaxations of the discrete set ${\cal X}$ (various other relaxations are possible as well). On the other hand, if one is interested in the exact solutions of (\ref{eq:ml1}), then the so-called Sphere-decoder (SD) algorithm from \cite{FinPhoSD85,HassVik05,JalOtt05} and the Branch-and-bound algorithms from \cite{StojnicBBSD08,StojnicBBSD05} are among the very best mathematically analyzable alternatives.

In the above mentioned \cite{Stojnicclupint19,Stojnicclupcmpl19,Stojnicclupplt19}, we introduced a completely new concept for attacking the optimization problem in (\ref{eq:ml1}) on the so-called \textbf{\emph{exact}} level. We called the resulting algorithm Controlled Loosening-up (\bl{\textbf{CLuP}}). It basically relies on the above mentioned very powerful \bl{\textbf{Random Duality Theory (RDT)}} concepts that we created in a long line of work \cite{StojnicCSetam09,StojnicCSetamBlock09,StojnicISIT2010binary,StojnicDiscPercp13,StojnicUpper10,StojnicGenLasso10,StojnicGenSocp10,StojnicPrDepSocp10,StojnicRegRndDlt10,Stojnicbinary16fin,Stojnicbinary16asym}, for solving and analyzing large classes of optimization problems (interestingly, among such problems are also the famous LASSO and SOCP alternatives of (\ref{eq:ml1}) which are among the most fundamental problems in  machine learning, compressed sensing, and statistics (see also, e.g. \cite{CheDon95,Tibsh96,DonMalMon10,BunTsyWeg07,vandeGeer08})). Before proceeding further we below recall on some of the CLuP's basic properties.

\subsection{Basic \bl{CLuP}}
\label{sec:clup}

One of the key features of the CLuP mechanism is its incredible simplicity. In its most basic form it amounts to the following iterative procedure. One basically starts with an initial vector, say $\x^{(0)}$ (which can be either deterministic or randomly generated from set ${\cal X}=\{-\frac{1}{\sqrt{n}},\frac{1}{\sqrt{n}}\}^n$) and proceeds iteratively:
\begin{eqnarray}
\x^{(i+1)}=\frac{\x^{(i+1,s)}}{\|\x^{(i+1,s)}\|_2} \quad \mbox{with}\quad \x^{(i+1,s)}=\mbox{arg}\min_{\x} & & -(\x^{(i)})^T\x  \nonumber \\
\mbox{subject to} & & \|\y-A\x\|_2\leq r\nonumber \\
&& \x\in \left [-\frac{1}{\sqrt{n}},\frac{1}{\sqrt{n}}\right ]^n. \label{eq:clup1}
\end{eqnarray}
As was extensively discussed in \cite{Stojnicclupint19,Stojnicclupcmpl19,Stojnicclupplt19}, one of the CLuP's most important components is the so-called radius $r$. As was further suggested in \cite{Stojnicclupint19,Stojnicclupcmpl19,Stojnicclupplt19}, $r$ was viewed as a multiple of $r_{plt}$, i.e. $r=r_{sc}r_{plt}\sqrt{n}$ where $r_{plt}$ is the $\sqrt{n}$ scaled radius that corresponds to the so-called polytope relaxation of the original problem in (\ref{eq:ml1}) (more on the polytope and other relaxations can be found in e.g. \cite{StojnicBBSD05,StojnicBBSD08} where we utilized it as the starting step in a branch-and-bound mechanism that we designed for finding the exact solution of (\ref{eq:ml1})). Although the whole CLuP concept seems incredibly simple a detailed theoretical analysis and a large set of numerical experiments in \cite{Stojnicclupint19,Stojnicclupcmpl19,Stojnicclupplt19} demonstrated that it performs remarkably well and easily approaches the so-called exact ML performance already on relatively small dimensions of order of several hundreds.

In Figure \ref{fig:fighighlightclup} we recall on the theoretical performance of the CLuP algorithm. As was mentioned earlier on multiple occasions, we will be interested in the so-called computationally hardest regimes, i.e. in the regimes where $\alpha<1$. In particular, in Figure \ref{fig:fighighlightclup} we set $\alpha=0.8$ and show the CLuP performance and how it compares to the corresponding ones of the well-known convex polytope relaxation heuristic (in \cite{Stojnicclupint19}, in addition to the Polytope-relaxation, we chose the Ball-relaxtion and the SDP-relaxation as well and the corresponding results related to those heuristics can be found there; as mentioned earlier, these heuristics are among the most popular and well known within the optimization theory and as such seemed as the most natural choice (see, e.g. \cite{GolVanLoan96Book,GroLovSch93Book,vanMaarWar00,GoeWill95} as well)).
\begin{figure}[htb]
\centering
\centerline{\epsfig{figure=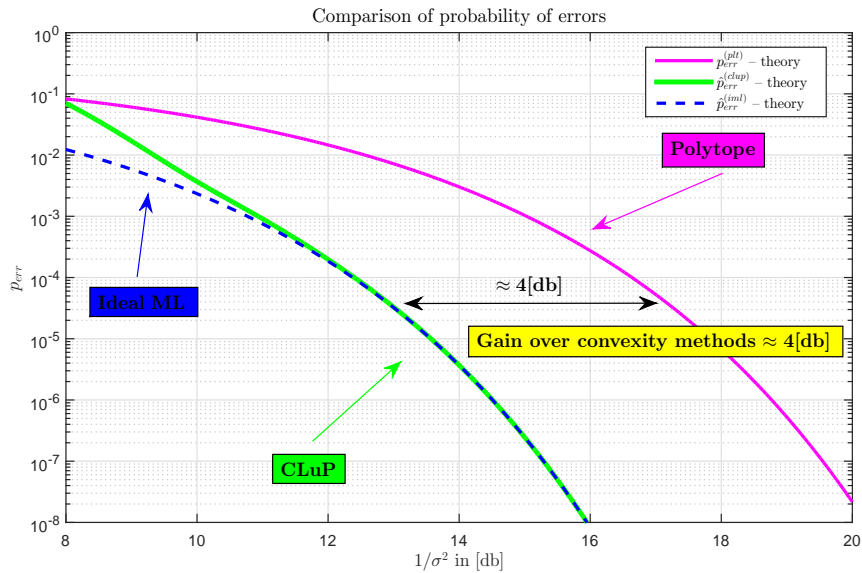,width=13.5cm,height=8cm}}
\caption{Comparison of $p_{err}$ as a function of $1/\sigma^2$; $\alpha=0.8$}
\label{fig:fighighlightclup}
\end{figure}
As Figure \ref{fig:fighighlightclup} (and the corresponding Figure 1 in \cite{Stojnicclupint19}) clearly indicate, CLuP is expected to substantially outperform all of the above mentioned convexity relaxation types of strategies. A large collection of numerical experiments that we conducted in \cite{Stojnicclupint19,Stojnicclupcmpl19,Stojnicclupplt19} confirmed that CLuP indeed approaches the ML performance while performing much better than the alternatives.

Looking at the structure of the CLuP one can immediately note that it is a very simple procedure. However, its theoretical analysis on the most demanding \emph{\textbf{per iteration level}} turned out to be a very serious challenge. Nonetheless, in \cite{Stojnicclupcmpl19} and \cite{Stojnicclupplt19} we were able to handle it and to dissect the overall CLuP performance through iterations into the tiniest of the details. Particularly remarkable was the fact that the theoretical predictions can be practically achieved often on the fifth decimal level already for a fairly small dimensions ($n$ of order of several hundreds was often sufficient). Besides the conclusions related to the performance in terms of probability of residual error (see Figure \ref{fig:fighighlightclup}), we also demonstrated that a rather small number of iterations was needed to achieve such a performance. In the most interesting regimes it was typically between $5$ and $15$ core iterations, where each core iteration essentially consisted of solving a quadratic program (for any fixed $i$, basically the one in (\ref{eq:clup1})). That of course implied that the overall complexity is basically the complexity of solving pretty much the simplest of quadratic programs. While quadratic programs are of course convex optimization problems and as such solvable through a host of well known techniques, their solving as the dimensions grow generically tends to require a cubic complexity. Running large scale examples of several tens/hundreds of thousands or millions (which are predicated to dominate in the future applications in the big data era) while doable in principle might then practically be a bit slow. Below we address this large dimensions phenomenon and create a CLuP program that is particularly tailored for large scale applications.

We will split the presentation into several parts. First we will introduce several Random Duality fundamental concepts that will be needed in the algorithms design. Then we will explain how such concepts can be utilized to eventually design practically useful algorithms. We will also present a solid set of numerical results, and observe a rather remarkable concurrence between the theoretical predictions and the practical realizations. Towards the end we will also provide a few concluding remarks.

\section{Algorithmic \bl{Random Duality Theory}}
\label{sec:algclup}

As mentioned quite a few times so far, to analyze the CLuP mechanism we in \cite{Stojnicclupint19,Stojnicclupcmpl19,Stojnicclupplt19} relied on the \bl{\textbf{Random Duality Theory}} and a long line of results that we created in a sequence of papers \cite{StojnicCSetam09,StojnicCSetamBlock09,StojnicISIT2010binary,StojnicDiscPercp13,StojnicUpper10,StojnicGenLasso10,StojnicGenSocp10,StojnicPrDepSocp10,StojnicRegRndDlt10,Stojnicbinary16fin,Stojnicbinary16asym}. Several different paths for analysis of various aspects of CLuP were provided in each of \cite{Stojnicclupint19,Stojnicclupcmpl19,Stojnicclupplt19}. Since we were interested in many specific features of CLuP (precision/probability of error, computational complexity, overall per iteration significant parameters tracking and so on) each analysis path was in a way tailored to emphasize particularities of such features. Below we follow one of such paths as it will be particularly convenient for showcasing the large scale CLuP capabilities.

We start things off by first reemphasizing a simple observation from  \cite{Stojnicclupint19,Stojnicclupcmpl19,Stojnicclupplt19} that CLuP is a trivially converging procedure. To that end we look at the ending optimization problem of the converging process
\begin{eqnarray}
\min_{\x} & & -\|\x\|_2  \nonumber \\
\mbox{subject to} & & \|\y-A\x\|_2\leq r\nonumber \\
&& \x\in \left [-\frac{1}{\sqrt{n}},\frac{1}{\sqrt{n}}\right ]^n. \label{eq:clup2}
\end{eqnarray}
Utilizing the expression for $\y$ from (\ref{eq:linsys1}) one can quickly transform (\ref{eq:clup2}) into
\begin{eqnarray}
\min_{\x} & & -\|\x\|_2  \nonumber \\
\mbox{subject to} & & \|[A \v]\begin{bmatrix}\x_{sol}-\x\\\sigma\end{bmatrix}\|_2\leq r\nonumber \\
&& \x\in \left [-\frac{1}{\sqrt{n}},\frac{1}{\sqrt{n}}\right ]^n. \label{eq:clup4}
\end{eqnarray}
As was the case in \cite{Stojnicclupint19,Stojnicclupcmpl19,Stojnicclupplt19}, we proceed in the standard \bl{\textbf{Random Duality Theory}} fashion and find it useful to work with two specific concentrating parameters, $c_1$ and $c_2$ defined in the following way
\begin{eqnarray}
c_2 & = & \|\x\|_2^2\nonumber \\
c_1 & = & (\x_{sol})^T\x. \label{eq:clup3}
\end{eqnarray}
One can then rewrite (\ref{eq:clup4}) as
\begin{eqnarray}
\min_{c_2\in[0,1]}\min_{\|\x\|_2^2=c_2} & & -\sqrt{c_2}  \nonumber \\
\mbox{subject to} & & \|[A \v]\begin{bmatrix}\x_{sol}-\x\\\sigma\end{bmatrix}\|_2\leq r\nonumber \\
&& \x\in \left [-\frac{1}{\sqrt{n}},\frac{1}{\sqrt{n}}\right ]^n. \label{eq:clup4ga}
\end{eqnarray}
Following further the basics of \bl{\textbf{RDT}} outlined in \cite{Stojnicclupint19,Stojnicclupcmpl19,Stojnicclupplt19} and after forming the standard Lagrangian we also have
\begin{eqnarray}
\min_{c_2\in[0,1]}\min_{\|\x\|_2^2=c_2} \max_{\gamma_1} & & -\sqrt{c_2} +\gamma_1\left (\mbox{max}_{\|\lambda\|_2=1}\lambda^T\left ([A \v]\begin{bmatrix}\x_{sol}-\x\\\sigma\end{bmatrix}\right )- r\right ) \nonumber \\
&& \x\in \left [-\frac{1}{\sqrt{n}},\frac{1}{\sqrt{n}}\right ]^n. \label{eq:clup5}
\end{eqnarray}
Continuing to follow further \cite{Stojnicclupint19} and relying on the $\gamma_1$'s concentration we arrive at the following variant of \cite{Stojnicclupint19}'s equations (8) and (44)
\begin{eqnarray}
\min_{c_2\in[0,1]}\max_{\gamma_1}\min_{\|\x\|_2^2=c_2} \max_{\|\lambda\|_2=1}  & & -\sqrt{c_2} +\gamma_1\lambda^T\left ([A \v]\begin{bmatrix}\x_{sol}-\x\\\sigma\end{bmatrix}\right )- \gamma_1r \nonumber \\
&& \x\in \left [-\frac{1}{\sqrt{n}},\frac{1}{\sqrt{n}}\right ]^n. \label{eq:clup5a}
\end{eqnarray}
One can then form the \bl{\textbf{Random dual}} and trivially redo all the steps from \cite{Stojnicclupint19}. The difference this time will be that we will not assume that signs of all elements of $\x_{sol}$ are equal. Instead we will assume that $\x_{sol}$ has $\rho n$ ($\rho\in[0,1]$) components equal to $\frac{1}{\sqrt{n}}$ and $(1-\rho) n$ components equal to $-\frac{1}{\sqrt{n}}$. Still all the machinery of \cite{Stojnicclupint19} can be reutilized to obtain the following as the optimizing objective of the \bl{\textbf{Random dual}}
\begin{equation}
\xi_{RD,\gamma_1}(\alpha,\sigma;c_2,c_1,\gamma,\nu)=-\sqrt{c_2}+\gamma_1(\sqrt{\alpha}\sqrt{1-2c_1+c_2+\sigma^2}+I_{22}-I_{1}+I_{21}-\nu c_1-\gamma c_2)-\gamma_1 r, \label{eq:discclup17}
\end{equation}
where $\gamma_1$ is $\sqrt{n}$ scaled version of $\gamma_1$ from (\ref{eq:clup5a}) and
$I_{22}$, $I_{1}$, and $I_{21}$ can be obtained as
\begin{eqnarray}
I_{22} &  = & \rho I_{22}(\gamma,\nu)+(1-\rho)I_{22}(\gamma,-\nu) \nonumber \\
I_{1}  &  = & \rho I_{1}(\gamma,\nu)+(1-\rho)I_{1}(\gamma,-\nu)  \nonumber \\
I_{21} &  = & \rho I_{21}(\gamma,\nu)+(1-\rho)I_{21}(\gamma,-\nu),
\label{eq:clupg16}
\end{eqnarray}
and $I_{22}(\gamma,\nu)$, $I_{1}(\gamma,\nu)$, and $I_{21}(\gamma,\nu)$ are similar to \cite{Stojnicclupint19}'s equation (22), i.e.
\begin{eqnarray}
I_{22}(\gamma,\nu) &  = & 0.5(\nu + \gamma)\erfc((\nu + 2\gamma)/\sqrt{2}) - exp(-0.5(\nu + 2\gamma)^2)/\sqrt{2\pi} \nonumber \\
I_{1}(\gamma,\nu)  &  = & (\sqrt{\pi/2}(\nu^2 + 1)\erf((2\gamma - \nu)/\sqrt{2}) + \sqrt{\pi/2}(\nu^2 + 1)\erf((2\gamma + \nu)/\sqrt{2})\nonumber \\& & + exp(-0.5(\nu + 2\gamma)^2) (\nu - 2\gamma)\nonumber  - exp(-0.5(\nu-2\gamma)^2)(\nu + 2\gamma))/(4\sqrt{2\pi}\gamma)  \nonumber \\
I_{21}(\gamma,\nu) &  = &  -0.5(\nu - \gamma)(\erf((\nu - 2\gamma)/\sqrt{2}) + 1) - exp(-0.5(\nu - 2\gamma)^2)/\sqrt{2\pi}.
\label{eq:clup16}
\end{eqnarray}
Utilizing \cite{Stojnicclupint19}'s Theorem 1, we then have as the object of interest the following optimization
\begin{equation}
\min_{c_2\in[0,1]}\min_{c_1\in[0,\sqrt{c_2}]}\max_{\gamma_1,\gamma,\nu}  \xi_{RD,\gamma_1}(\alpha,\sigma;c_2,c_1,\gamma,\nu). \label{eq:clupg5a}
\end{equation}
To handle the above optimization we look at the stationary points. We start with the derivatives with respect to $c_1$ and $c_2$. First we have
\begin{equation}\label{eq:derc1cclupg1}
  \frac{d\xi_{RD,\gamma_1}(\alpha,\sigma;c_2,c_1,\gamma,\nu)}{d c_1}=
  -\frac{\sqrt{\alpha}}{\sqrt{1-2c_1+c_2+\sigma^2}}-\nu=0,
\end{equation}
which gives
\begin{equation}\label{eq:derc1cclupg2}
 \nu = -\frac{\sqrt{\alpha}}{\sqrt{1-2c_1+c_2+\sigma^2}}.
\end{equation}
Then we also find
\begin{equation}\label{eq:derc2cclupg1}
  \frac{d\xi_{RD,\gamma_1}(\alpha,\sigma;c_2,c_1,\gamma,\nu)}{d c_2}=
  -\frac{1}{2\sqrt{c_2}}+\frac{\gamma_1\sqrt{\alpha}}{2\sqrt{1-2c_1+c_2+\sigma^2}}-\gamma_1\gamma=-\frac{1}{2\sqrt{c_2}}-\frac{\gamma_1\nu}{2}-\gamma_1\gamma=0,
\end{equation}
which gives
\begin{equation}\label{eq:derc2cclupg2}
  \gamma_1=\frac{1}{2\sqrt{c_2}(-\nu/2-\gamma)}.
\end{equation}
The derivative with respect to $\gamma_1$ gives
\begin{equation}\label{eq:dergamma1cclupg1}
  \frac{d\xi_{RD,\gamma_1}(\alpha,\sigma;c_2,c_1,\gamma,\nu)}{d\gamma_1}=\sqrt{\alpha}\sqrt{1-2c_1+c_2+\sigma^2}+I_{22}-I_{1}+I_{21}-\nu c_1-\gamma c_2-r.
\end{equation}
For the derivative with respect to $\nu$ we have
\begin{equation}\label{eq:dernucclupg1}
  \frac{d\xi_{RD,\gamma_1}(\alpha,\sigma;c_2,c_1,\gamma,\nu)}{d\nu}=\rho\frac{dI(\gamma,\nu)}{d\nu}+(1-\rho)\frac{dI(\gamma,-\nu)}{d\nu}-c_1,
\end{equation}
where
\begin{equation}\label{eq:dernucclupg2}
I(\gamma,\nu)=I_{22}(\gamma,\nu)-I_{1}(\gamma,\nu)+I_{21}(\gamma,\nu),
\end{equation}
and consequently
\begin{equation}\label{eq:dernucclupg3}
\frac{dI(\gamma,\nu)}{d\nu}=\frac{dI_{22}(\gamma,\nu)}{d\nu}-\frac{dI_{1}(\gamma,\nu)}{d\nu}+\frac{dI_{21}(\gamma,\nu)}{d\nu}.
\end{equation}
Utilizing (\ref{eq:clup16}) we obtain
\begin{eqnarray}\label{eq:dernucclupg4}
\frac{dI_{22}(\gamma,\nu)}{d\nu}&=&\frac{\gamma}{\sqrt{2\pi}}\exp(-0.5(2\gamma + \nu)^2) + 0.5\erfc((2\gamma + \nu)/\sqrt{2})\nonumber \\
\frac{dI_{1}(\gamma,\nu)}{d\nu}&=&(-\sqrt{2\pi}\nu\erf((\nu - 2\gamma)/\sqrt{2}) + \sqrt{2\pi}\nu\erf((2\gamma +\nu)/\sqrt{2})\nonumber \\
& &- (\nu^2 + 1)\exp(-1/2(\nu - 2 \gamma)^2) + (\nu^2 + 1)\exp(-1/2 (2\gamma + \nu)^2)\nonumber \\
& &+ \exp(-1/2 (\nu - 2 \gamma)^2) (\nu - 2\gamma)(2\gamma +\nu) - \exp(-1/2(2\gamma + \nu)^2)(\nu - 2 \gamma)(2\gamma + \nu)\nonumber \\
 & &- \exp(-1/2(2\gamma - \nu)^2) + \exp(-1/2 (2 \gamma + \nu)^2))/(4\sqrt{2\pi}\gamma)\nonumber\\
\frac{dI_{21}(\gamma,\nu)}{d\nu}&=&-\frac{\gamma}{\sqrt{2\pi}}\exp(-0.5(-2\gamma + \nu)^2) - 0.5\erfc((2\gamma - \nu)/\sqrt{2}).
\end{eqnarray}
One can the recompute all of the above derivatives for $-\nu$ or utilize the following
\begin{equation}\label{eq:dernucclupg5}
\frac{dI(\gamma,-\nu)}{d\nu}=\frac{dI(\gamma,-\nu)}{d(-\nu)}\frac{d(-\nu)}{d\nu}=-\frac{dI(\gamma,\nu)}{d\nu}|_{\nu=-\nu}.
\end{equation}
For the derivative with respect to $\gamma$ we have
\begin{equation}\label{eq:dergammacclupg1}
  \frac{d\xi_{RD,\gamma_1}(\alpha,\sigma;c_2,c_1,\gamma,\nu)}{d\gamma}=\rho\frac{dI(\gamma,\nu)}{d\gamma}+(1-\rho)\frac{dI(\gamma,-\nu)}{d\gamma}-c_2,
\end{equation}
where similarly to (\ref{eq:dernucclupg3})
\begin{equation}\label{eq:dergammacclupg3}
\frac{dI(\gamma,\nu)}{d\gamma}=\frac{dI_{22}(\gamma,\nu)}{d\gamma}-\frac{dI_{1}(\gamma,\nu)}{d\gamma}+\frac{dI_{21}(\gamma,\nu)}{d\gamma}.
\end{equation}
Utilizing again (\ref{eq:clup16}) we obtain
\begin{eqnarray}\label{eq:dergammacclupg4}
\frac{dI_{22}(\gamma,\nu)}{d\gamma}&=&\gamma\sqrt{\frac{2}{\pi}}\exp(-0.5(2\gamma + \nu)^2) + 0.5\erfc((2\gamma + \nu)/\sqrt{2})\nonumber \\
\frac{dI_{1}(\gamma,\nu)}{d\gamma}&=&(2(\nu^2 + 1)\exp(-1/2(\nu - 2\gamma)^2) + 2(\nu^2 + 1)\exp(-1/2 (2\gamma + \nu)^2) \nonumber \\
& & - 2\exp(-1/2(2\gamma - \nu)^2) (\nu - 2\gamma) (2\gamma + \nu) - 2\exp(-1/2 (2\gamma + \nu)^2) (\nu - 2\gamma)(2\gamma + \nu) \nonumber \\
& & - 2\exp(-1/2(2\gamma - \nu)^2) - 2\exp(-1/2(2\gamma + \nu)^2))/(4\sqrt{2\pi}\gamma)\nonumber \\
& &- (-\sqrt{\pi/2}(\nu^2 + 1)\erf((\nu - 2\gamma)/\sqrt{2}) + \sqrt{\pi/2}(\nu^2 + 1)\erf((2\gamma +\nu)/\sqrt{2})\nonumber \\
& & + \exp(-1/2 (2\gamma + \nu)^2)(\nu - 2\gamma) - \exp(-1/2(2\gamma - \nu)^2) (2\gamma + \nu))/(4\sqrt{2\pi}\gamma^2)\nonumber\\
\frac{dI_{21}(\gamma,\nu)}{d\gamma}&=&\gamma\sqrt{\frac{2}{\pi}}\exp(-0.5(-2\gamma + \nu)^2) + 0.5\erfc((2\gamma - \nu)/\sqrt{2}).
\end{eqnarray}
Finally, one also has
\begin{equation}\label{eq:dergammacclupg5}
\frac{dI(\gamma,-\nu)}{d\gamma}=\frac{dI(\gamma,\nu)}{d\gamma}|_{\nu=-\nu}.
\end{equation}
We summarize the above considerations in the following theorem.
\begin{theorem}(\bl{\textbf{Random dual}} -- stationary points)
Let $\xi_{RD,\gamma_1}(\alpha,\sigma;c_2,c_1,\gamma,\nu)$ be as in (\ref{eq:discclup17}). Then its stationary points are all solutions to the following system of equations:
\begin{eqnarray}
  \frac{d\xi_{RD,\gamma_1}(\alpha,\sigma;c_2,c_1,\gamma,\nu)}{d\nu} & = & \rho\frac{dI(\gamma,\nu)}{d\nu}+(1-\rho)\frac{dI(\gamma,-\nu)}{d\nu}-c_1=0\nonumber \\
    \frac{d\xi_{RD,\gamma_1}(\alpha,\sigma;c_2,c_1,\gamma,\nu)}{d\gamma} & = &\rho\frac{dI(\gamma,\nu)}{d\gamma}+(1-\rho)\frac{dI(\gamma,-\nu)}{d\gamma}-c_2=0 \nonumber \\
      \frac{d\xi_{RD,\gamma_1}(\alpha,\sigma;c_2,c_1,\gamma,\nu)}{d\gamma_1} & = & \sqrt{\alpha}\sqrt{1-2c_1+c_2+\sigma^2}+I_{22}-I_{1}+I_{21}-\nu c_1-\gamma c_2-r=0 \nonumber \\
      \nu & = & -\frac{\sqrt{\alpha}}{\sqrt{1-2c_1+c_2+\sigma^2}} \nonumber \\
        \gamma_1 & = & \frac{1}{2\sqrt{c_2}(-\nu/2-\gamma)},
\end{eqnarray}
where $\frac{dI(\gamma,\nu)}{d\nu}$ and $\frac{dI(\gamma,-\nu)}{d\nu}$ are as given in (\ref{eq:dernucclupg3})-(\ref{eq:dernucclupg5}), $\frac{dI(\gamma,\nu)}{d\gamma}$ and $\frac{dI(\gamma,-\nu)}{d\gamma}$ are as given in (\ref{eq:dergammacclupg3})-(\ref{eq:dergammacclupg5}), and $I_{22}$, $I_{1}$, and $I_{21}$ are as given in (\ref{eq:clupg16})-(\ref{eq:clup16}).
\label{thm:cluprd1}
\end{theorem}
\begin{proof}
Follows as a consequence of the above analysis and a host of fundamental RDT properties developed in \cite{StojnicCSetam09,StojnicISIT2010binary,StojnicDiscPercp13,StojnicGenLasso10,StojnicGenSocp10,StojnicPrDepSocp10,StojnicRegRndDlt10} and in particular in  \cite{Stojnicclupint19,Stojnicclupcmpl19,Stojnicclupplt19}.
\end{proof}
When $c_2=1$ the above needs to be slightly adjusted. However, such a scenario will not appear as of much interest below and we skip this easy exercise and overloading the overall presentation with such details. Below we discuss how the above machinery can be utilized to design efficient optimization algorithms.

\subsection{Transforming \bl{constrained} into \bl{unconstrained} optimization}
\label{sec:consuncons}

When we created the Random Duality Theory in \cite{StojnicCSetam09,StojnicCSetamBlock09,StojnicISIT2010binary,StojnicDiscPercp13,StojnicUpper10,StojnicGenLasso10,StojnicGenSocp10,StojnicPrDepSocp10,StojnicRegRndDlt10,Stojnicbinary16fin,Stojnicbinary16asym}, it was immediately clear that it has quite a few great features. Two of them were among the most prominent.

\tcbset{colback=yellow!95!white,colframe=blue!95!white,fonttitle=\bfseries}
\begin{tcolorbox}[title=Two of the most prominent \textbf{RDT} features]
\begin{itemize}
  \item \textbf{\emph{\bl{RDT} has a massive power in providing the exact characterization of various random structures.}}
\item \textbf{\emph{The machinery of \bl{RDT} established in \cite{StojnicCSetam09,StojnicCSetamBlock09,StojnicISIT2010binary,StojnicDiscPercp13,StojnicUpper10,StojnicGenLasso10,StojnicGenSocp10,StojnicPrDepSocp10,StojnicRegRndDlt10,Stojnicbinary16fin,Stojnicbinary16asym} also provided large scale algorithms of an almost unparallel computational efficiency.}}
\end{itemize}.\vspace{-.3in}
\end{tcolorbox}
The first feature basically implied that an exact performance analysis of optimization problems became doable with the analysis precision often reaching fourth or fifth decimal even for problems of relatively small dimensions of few hundreds (even though the RDT analysis was conceptually designed for basically infinitely dimensional problems). On the other hand, the second feature was the theory's ability to provide ways to create practical, computationally efficient algorithms than can actually achieve the above mentioned analytically exactly predictable performance. Below we briefly revisit the fundamental ideas from \cite{StojnicCSetam09,StojnicCSetamBlock09,StojnicISIT2010binary,StojnicDiscPercp13,StojnicUpper10,StojnicGenLasso10,StojnicGenSocp10,StojnicPrDepSocp10,StojnicRegRndDlt10,Stojnicbinary16fin,Stojnicbinary16asym} that we used over the years for creating very fast large scale algorithms and discuss how they work within the CLuP context. First we once again recall the ending optimization of the CLuP procedure
\begin{eqnarray}
\min_{\x} & & -\|\x\|_2  \nonumber \\
\mbox{subject to} & & \|\y-A\x\|_2\leq r\nonumber \\
&& \x\in \left [-\frac{1}{\sqrt{n}},\frac{1}{\sqrt{n}}\right ]^n. \label{eq:ctucclup2}
\end{eqnarray}
Solving this problem on a large scale even in the hands of the very best optimization theorist might be quite challenging. There are of course many reasons for that. For example, even if one can find a procedure that would lead towards the solution, it remains almost inconceivable how handling the set of constraints can be so to say computationally avoided which often is pretty much a necessary step needed to create algorithms powerful enough to be capable of handling very large dimensions. Instead of relying on the standard optimization techniques we of course again focus on the \bl{\textbf{Random Duality Theory}}. We first write the Lagrangian of the above primal
\begin{eqnarray}
\min_{\x}\max_{\gamma_1\geq 0} & & -\|\x\|_2+\gamma_1(\|\y-A\x\|_2- r)  \nonumber \\
\mbox{subject to} & & \x\in \left [-\frac{1}{\sqrt{n}},\frac{1}{\sqrt{n}}\right ]^n. \label{eq:ctucclup4}
\end{eqnarray}
Now, as is trivially known even for the optimization beginners, the above problem does remove the critical set of constraints and seemingly transforms the constrained primal optimization into practically speaking an unconstrained one (the individual constraints are typically much easier to handle and we don't view them as particularly problematic in the discussion here). Of course, that comes at the expense of having the original single optimization being transformed into a double one. As expected, the above Lagrangian then doesn't really help much with the idea of avoiding the existence of critical constraints. This is all well known and generically true for pretty much any optimization problem. However, such a reasoning remains valid only so long until one is capable of finding a quick way of computing $\gamma_1$. That is exactly where the Random Duality Theory comes in place. Utilizing the RDT concepts that we developed in \cite{StojnicCSetam09,StojnicCSetamBlock09,StojnicISIT2010binary,StojnicDiscPercp13,StojnicUpper10,StojnicGenLasso10,StojnicGenSocp10,StojnicPrDepSocp10,StojnicRegRndDlt10,Stojnicbinary16fin,Stojnicbinary16asym} one can determine $\gamma_1$ for pretty much any problem that falls within the frame of RDT. When it comes in particular to the problems of our interest here, then the above machinery developed based on RDT and summarized in the above theorem provides the way to determine $\gamma_1$. That basically means that through RDT we have automatically solved one of the biggest obstacles in creating large scale capable algorithms. Since the problem is conceptually solved the only things that are left to be done are technical realizations and we discuss them in a separate section below. There are many ways how one can go about the technical realizations. Since this is the introductory paper, we will focus on some of the most basic implementations and in some of our companion papers we will discuss more advanced optimization tools that can also be utilized.

\section{Large scale CLuP}
\label{sec:cluplargesc}

As mentioned above, once the results of Theorem \ref{thm:cluprd1} are available, there are many practical ways how one can handle (\ref{eq:ctucclup4}). We will consider a mixed stationary points-constraints satisfaction type of implementations. To that end we start with the following quick observation
\begin{eqnarray}
\min_{\x} & & \xi_{LS}  \nonumber \\
\mbox{subject to} & & \xi_{LS}=-\|\x\|_2+\hat{\gamma}_1(\|\y-A\x\|_2- r)  \nonumber \\
& & \x\in \left [-\frac{1}{\sqrt{n}},\frac{1}{\sqrt{n}}\right ]^n, \label{eq:LSclup1}
\end{eqnarray}
where $\hat{\gamma}_1$ is obtained through the machinery of Theorem \ref{thm:cluprd1} (and scaled by $\sqrt{n}$). We then also have
\begin{equation}
\frac{d\xi_{LS}}{d\x} =-\frac{\x}{\|\x\|_2}+\hat{\gamma}_1\frac{-A^T(\y-A\x)}{\|\y-A\x\|_2}. \label{eq:LSclup2}
\end{equation}
After setting the derivative to zero we further have from (\ref{eq:LSclup2})
\begin{equation}
\frac{d\xi_{LS}}{d\x}=0 \Longleftrightarrow -\x\|\y-A\x\|_2-\hat{\gamma}_1A^T\y\|\x\|_2+\hat{\gamma}_1A^TA\x\|\x\|_2=0. \label{eq:LSclup3}
\end{equation}
There are many ways how one can solve the above equation. We leave more complex choices for some of our companion papers, here we choose a simple scaled regularized contraction. Basically, we set
\begin{eqnarray}
& &  -\x\|\y-A\x\|_2-\hat{\gamma}_1A^T\y\|\x\|_2+\hat{\gamma}_1A^TA\x\|\x\|_2=0 \nonumber \\
&\Longleftrightarrow & c_{q,2}\x-\x\|\y-A\x\|_2-c_{q,2}\x-\hat{\gamma}_1A^T\y\|\x\|_2+\hat{\gamma}_1A^TA\x\|\x\|_2=0\nonumber \\
&\Longleftrightarrow & \x(c_{q,2}-\|\y-A\x\|_2)=c_{q,2}\x+\hat{\gamma}_1A^T\y\|\x\|_2-\hat{\gamma}_1A^TA\x\|\x\|_2 \nonumber \\
&\Longleftrightarrow & \x=\frac{c_{q,2}\x+\hat{\gamma}_1A^T\y\|\x\|_2-\hat{\gamma}_1A^TA\x\|\x\|_2}{c_{q,2}-\|\y-A\x\|_2}, \label{eq:LSclup4}
\end{eqnarray}
where $c_{q,2}$ is an appropriately chosen constant. Taking into account the individual box constraints one can then established the following contraction
\begin{equation}
\x^{(i+1,r)}=\frac{c_{q,2}\x^{(i)}+\hat{\gamma}_1A^T\y\|\x^{(i)}\|_2-\hat{\gamma}_1A^TA\x^{(i)}\|\x^{(i)}\|_2}{c_{q,2}-\|\y-A\x^{(i)}\|_2},
\x^{(i+1)}=\begin{cases}
  -\frac{1}{\sqrt{n}}, & \mbox{if } \x^{(i+1,r)}\leq -\frac{1}{\sqrt{n}} \\
  \x^{(i+1,r)}, & \mbox{if } -\frac{1}{\sqrt{n}}\leq \x^{(i+1,r)}\leq \frac{1}{\sqrt{n}} \\
  \frac{1}{\sqrt{n}}, & \mbox{otherwise}.
\end{cases}\label{eq:LSclup5}
\end{equation}
However, depending on the implementation one can go a step further and utilize the RDT to an even larger degree. Namely, after observing that in the limit $\|\x^{(i)}\|_2\rightarrow \sqrt{\hat{c}_2}$ and $\|\y-A\x^{(i)}\|_2\rightarrow r_{sc}r_{plt}\sqrt{n}$, where, similarly to $\gamma_1$, $\hat{c}_2$ is also obtained through the machinery of Theorem \ref{thm:cluprd1}, one can also have the following contraction
\begin{equation}
\x^{(i+1,r)}=\frac{c_{q,2}\x^{(i)}+\hat{\gamma}_1A^T\y\sqrt{\hat{c}_2}-\hat{\gamma}_1A^TA\x^{(i)}\sqrt{\hat{c}_2}}{c_{q,2}-r_{sc}r_{plt}\sqrt{n}},
\x^{(i+1)}=\begin{cases}
  -\frac{1}{\sqrt{n}}, & \mbox{if } \x^{(i+1,r)}\leq -\frac{1}{\sqrt{n}} \\
  \x^{(i+1,r)}, & \mbox{if } -\frac{1}{\sqrt{n}}\leq \x^{(i+1,r)}\leq \frac{1}{\sqrt{n}} \\
  \frac{1}{\sqrt{n}}, & \mbox{otherwise}.
\end{cases}\label{eq:LSclup6}
\end{equation}
One can then formalize the above procedure and run it. There are a couple of things that one should immediately note. First, the complexity per iteration is remarkably low. Namely, there is only one matrix-vector multiplication which amounts to overall per iteration complexity of $mn$ (or in scaling terms, it amounts to a quadratic complexity per iteration). Also, one can run it successively. Namely, after getting certain output $\x^{(CLuP)}$ after say $i_{max}$ iterations, one can rerun it with the $\x^{(CLuP)}$ as the $\x^{(0)}$. One can also repeat rerunning as many times as desired. We found this as particularly useful in decreasing the number of running iterations. In the following section we discuss these and a few other algorithm's features in a bit more details.

\begin{algorithm}[t]
\caption{Large scale CLuP -- achieving \textbf{exact} ML in polynomial time)}

{\bf Input:} Received vector $\y \in \mR^m$, system matrix $A\in \mR^{m\times n}$, radius $r$, $\hat{\gamma}_1$, $\hat{c}_2$, $c_{q,2}$, starting unknown vector $\x^{(0)}\in\mR^n$, maximum number of iterations $i_{max}$, desired converging precision $\delta_{min}$.[\bl{$\mbox{CLuP}(\y,A,r,\hat{\gamma}_1, \hat{c}_2, c_{q,2},\x^{(0)},i_{max},\delta)$}] \\
{\bf Output:} Estimated vector $\x^{(i)} \in \mR^n$ and its discretized  variant $\x^{(CLuP)}$.[\bl{$\x^{(i)},\x^{(CLuP)}$}]

\begin{algorithmic}[1]

\STATE Initialize the convergence gap and the iteration counter, $\delta\leftarrow 10^{10}$ and $i\leftarrow 0$

\STATE Set $c_{2,\|\|}^{(0)}\leftarrow \delta^2$

\WHILE{$i+1\leq i_{max}$ and/or $\delta\geq\delta_{min}$}

\STATE Obtain $\x^{(i+1,r)}$
\begin{equation*}\label{eq:alg1}
\x^{(i+1,r)}=\frac{c_{q,2}\x^{(i)}+\hat{\gamma}_1A^T\y\sqrt{\hat{c}_2}-\hat{\gamma}_1A^TA\x^{(i)}\sqrt{\hat{c}_2}}{c_{q,2}-r_{sc}r_{plt}\sqrt{n}}.
\end{equation*}

\STATE Set
\begin{equation*}\label{eq:alg2}
\x^{(i+1)}=\begin{cases}
  -\frac{1}{\sqrt{n}}, & \mbox{if } \x^{(i+1,r)}\leq -\frac{1}{\sqrt{n}} \\
  \x^{(i+1,r)}, & \mbox{if } -\frac{1}{\sqrt{n}}\leq \x^{(i+1,r)}\leq \frac{1}{\sqrt{n}} \\
  \frac{1}{\sqrt{n}}, & \mbox{otherwise}.
\end{cases}
\end{equation*}

\STATE Set $c_{2,\|\|}^{(i+1)}\leftarrow \mbox{sign}(\x^{(i)})^T\mbox{sign}(\x^{(i+1)})$

\STATE Set $\delta\leftarrow |1-c_{2,\|\|}^{(i+1)}|$

\STATE Update the iteration counter $i\leftarrow i+1$

\ENDWHILE

\STATE $\x^{(CLuP)}\leftarrow \frac{1}{\sqrt{n}}\mbox{sign}(\x^{(i)})$.

\end{algorithmic}

\end{algorithm}

\section{Numerical results}
\label{sec:numres}

In this section we present a large set of numerical results that relate to both, the theoretical predictions and the practical algorithm running. To start things off, we in Figure \ref{fig:simLSclupprerr} show the results that can be obtained through the rerunning of the above large scale CLuP (we will sometimes refer to it as $\text{CLuP}^{r_0}$). As the theory suggests, the $\text{CLuP}^{r_0}$ achieves performance almost identical to the ML and substantially better (around 4[db]) than the corresponding polytope convex relaxation. In Figure \ref{fig:simLSclupprerr}, we in addition to the CLuP and polytope relaxation curves show the curve that corresponds to the performance of an ideal detection on a single Gaussian channel with binary signals and appropriately scaled SNR (we refer to this type of performance and the corresponding curve as the Ideal ML). The exact formula is trivially given after the SNR rescaling as
\begin{equation}\label{eq:numres1}
  p_{err}^{(iml)}=\frac{1}{2}\erfc(\sqrt{\alpha/2}/\sqrt{10^{-(1/\sigma^2)[\text{db}]/10}}).
\end{equation}

\begin{figure}[htb]
\centering
\centerline{\epsfig{figure=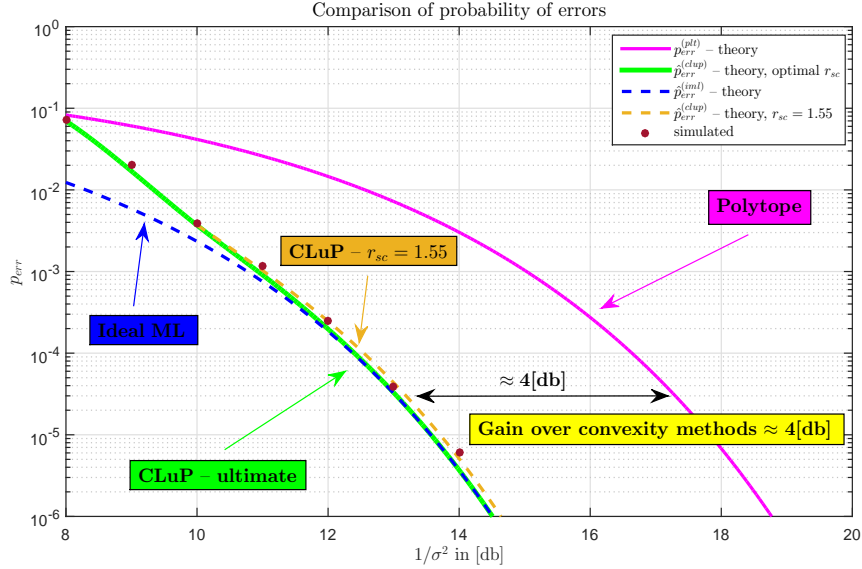,width=13.5cm,height=8cm}}
\caption{Comparison of $p_{err}$ as a function of $1/\sigma^2$; $\alpha=0.8$}
\label{fig:simLSclupprerr}
\end{figure}
We do recall from \cite{Stojnicclupint19}, that the ultimate CLuP's performance is determined for the worst case probability of error. The numerical values for the corresponding $r_{sc}$ and the resulting probabilities of error are given in Table \ref{tab:tabclup1}. The probabilities of error from Table \ref{tab:tabclup1} correspond to the green curve in Figure \ref{fig:simLSclupprerr}.
\begin{table}[h]
\caption{Numerical values for $r_{sc}$ and $\hat{p}_{err}^{(clup)}$ that correspond to the data in Figure \ref{fig:simLSclupprerr} (green curve)} \vspace{.1in}
\hspace{-0in}\centering
\footnotesize{
\begin{tabular}{||c||c|c|c|c|c|c|c||}\hline\hline
$ 1/\sigma^2 $[db] & $8  $ & $9  $ & $10  $ & $11  $ & $12  $ & $13  $ & $14  $  \\ \hline\hline
$r_{sc}$ & $1.10  $ & $1.41  $ & $1.56  $ & $1.611  $ & $1.627  $ & $1.632  $ & $1.633  $  \\ \hline
$\hat{p}_{err}^{(clup)}$ & $6.98e-02  $ & $1.70e-02  $ & $3.69e-03  $ & $9.09e-04  $ & $1.97e-04  $ & $3.29e-05  $ & $3.70e-06  $  \\ \hline\hline
\end{tabular}}
\label{tab:tabclup1}
\end{table}
We conducted numerical simulations with $r_{sc}$ values that are a bit below the optimal ones as it is a bit easier to achieve better and faster concentrations of all critical system parameters. For concreteness, we in Table \ref{tab:tabclup2} show the exact values of $r_{sc}$ that were simulated and the theoretical predictions for $\hat{p}_{err}^{(clup)}$ that one can obtain for such values. We also plot (as a light brown curve) in Figure \ref{fig:simLSclupprerr} the theoretical predictions that can be obtained for a fixed value, $r_{sc}=1.55$. Such a value is either equal (for $1/\sigma^2=12-14$[db]) or very close (for $1/\sigma^2=10-11$[db]) to the simulated values. We would like to emphasize that one of the very best features of the original CLuP from \cite{Stojnicclupint19} is that one does not necessarily need to know the SNR to be able to run it. Moreover, one can run it in a so to say universal way, i.e. one can choose a fixed $r_{sc}$ and obtain performance almost identical to the ML prediction. We showcased such CLuP abilities on multiple occasions in \cite{Stojnicclupint19,Stojnicclupcmpl19,Stojnicclupplt19}. Here, while we are in a way more interested in CLuP's ultimate behavior the curve for $r_{sc}=1.55$ in Figure \ref{fig:simLSclupprerr} and the results from Table \ref{tab:tabclup2} demonstrate that in the most interesting SNR regimes even the ultimate CLuP's behavior can be very closely approached with a universal type of running.
\begin{table}[h]
\caption{Numerical values for $r_{sc}$ and $\hat{p}_{err}^{(clup)}$ that correspond to the simulated data in Figure \ref{fig:simLSclupprerr}} \vspace{.1in}
\hspace{-0in}\centering
\footnotesize{
\begin{tabular}{||c||c|c|c|c|c|c|c||}\hline\hline
$ 1/\sigma^2 $[db] & $8  $ & $9  $ & $10  $ & $11  $ & $12  $ & $13  $ & $14  $ \\ \hline\hline
$r_{sc}$ & $1.10  $ & $1.35  $ & $1.56  $ & $1.5  $ & $1.55  $ & $1.55  $ & $1.55  $ \\ \hline
$\hat{p}_{err}^{(clup)}$ & $6.98e-02  $ & $1.79e-02  $ & $3.69e-03  $ & $1.19e-03  $ & $2.46e-04  $ & $4.35e-05  $ & $5.11e-06  $  \\ \hline\hline
\end{tabular}}
\label{tab:tabclup2}
\end{table}
Finally, in Table \ref{tab:tabclup3} we show in parallel both the theoretical and the simulated values for all critical parameters ($c_2$, $c_1$, $\hat{p}_{err}^{(clup)}$, and $r$). As explained earlier, we reran $\text{CLuP}^{r_0}$ with $n=2000$ and $i_{max}=300$ (we found $c_{q,2}=5\sqrt{n}$ to be a solid starting point; for $1/\sigma^2 \in\{9,10\}$[db] the procedure was a bit different and we will discuss it separately below). One can observe a very strong agreement between what the theory predicts and what one can get through the numerical simulations.
\begin{table}[h]
\caption{\textbf{Theoretical}/\bl{\textbf{simulated}} values for $c_2$, $c_1$, $\hat{p}_{err}^{(clup)}$, and $r$ in Figure \ref{fig:simLSclupprerr} ($n=2000$)} \vspace{.1in}
\hspace{-0in}\centering
\footnotesize{
\begin{tabular}{||c||c||c|c||c|c||c|c||c|c||}\hline\hline
$ 1/\sigma^2 $[db] & $\hat{\gamma}_1\sqrt{n}  $ & $c_2$ & $c_2$ & $c_1$ & $c_1$ & $\hat{p}_{err}^{(clup)}$ & $\hat{p}_{err}^{(clup)}$ & $\frac{r}{\sqrt{n}}$ & $\frac{r}{\sqrt{n}}$ \\ \hline\hline
$8 $ & $\mathbf{1.8022 }$ & $\mathbf{0.8325 }$ & $\bl{\mathbf{0.8335 }}$ & $\mathbf{0.8120 }$ & $\bl{\mathbf{0.8062 }}$ & $\mathbf{6.9878e-02 }$ & $\bl{\mathbf{7.3373e-02 }}$ & $\mathbf{0.2401 }$ & $\bl{\mathbf{0.2401 }}$ \\ \hline
$9 $ & $\mathbf{0.7916 }$ & $\mathbf{0.9432 }$ & $\bl{\mathbf{0.9435 }}$ & $\mathbf{0.9437 }$ & $\bl{\mathbf{0.9398 }}$ & $\mathbf{1.7933e-02 }$ & $\bl{\mathbf{2.0148e-02 }}$ & $\mathbf{0.2624 }$ & $\bl{\mathbf{0.2620 }}$ \\ \hline
$10 $ & $\mathbf{0.4657 }$ & $\mathbf{0.9910 }$ & $\bl{\mathbf{0.9912 }}$ & $\mathbf{0.9898 }$ & $\bl{\mathbf{0.9896 }}$ & $\mathbf{3.6882e-03 }$ & $\bl{\mathbf{3.8263e-03 }}$ & $\mathbf{0.2702 }$ & $\bl{\mathbf{0.2701 }}$ \\ \hline
$11 $ & $\mathbf{0.5816 }$ & $\mathbf{0.9815 }$ & $\bl{\mathbf{0.9816 }}$ & $\mathbf{0.9871 }$ & $\bl{\mathbf{0.9872 }}$ & $\mathbf{1.1872e-03 }$ & $\bl{\mathbf{1.1693e-03 }}$ & $\mathbf{0.2316 }$ & $\bl{\mathbf{0.2315 }}$ \\ \hline
$12 $ & $\mathbf{0.5119 }$ & $\mathbf{0.9905 }$ & $\bl{\mathbf{0.9905 }}$ & $\mathbf{0.9938 }$ & $\bl{\mathbf{0.9939 }}$ & $\mathbf{2.4610e-04 }$ & $\bl{\mathbf{2.4485e-04 }}$ & $\mathbf{0.2132 }$ & $\bl{\mathbf{0.2129 }}$ \\ \hline
$13 $ & $\mathbf{0.5243 }$ & $\mathbf{0.9913 }$ & $\bl{\mathbf{0.9913 }}$ & $\mathbf{0.9947 }$ & $\bl{\mathbf{0.9948 }}$ & $\mathbf{4.3485e-05 }$ & $\bl{\mathbf{3.8916e-05 }}$ & $\mathbf{0.1901 }$ & $\bl{\mathbf{0.1900 }}$ \\ \hline
$14 $ & $\mathbf{0.5348 }$ & $\mathbf{0.9920 }$ & $\bl{\mathbf{0.9920 }}$ & $\mathbf{0.9954 }$ & $\bl{\mathbf{0.9954 }}$ & $\mathbf{5.1122e-06 }$ & $\bl{\mathbf{6.1052e-06 }}$ & $\mathbf{0.1694 }$ & $\bl{\mathbf{0.1693 }}$ \\ \hline\hline
\end{tabular}}
\label{tab:tabclup3}
\end{table}

\subsection{Rephasing}
\label{sec:rephasing}

As mentioned above, for $1/\sigma^2 =9$ or $10$[db] we ran a slightly different procedure. Namely, although one can achieve the theoretical predictions with a bit larger $n$ for the values given in the above tables, we found a bit helpful for lower dimensions to do the so-called CLuP \textbf{rephasing}. That means that we first reran $\text{CLuP}^{r_0}$ with $n=2000$ and $i_{max}=300$ for ceratin values of parameters and then repeated the same type of rerunning for a different set of parameters (basically in the second rerunning the parameters were exactly as specified in Table \ref{tab:tabclup3}). The parameters that we selected for both rerunning (we refer to them as phase 0 and phase 1) are shown in Table \ref{tab:tabclup4} for $1/\sigma^2=9$[db] and in Table \ref{tab:tabclup5} for $1/\sigma^2=10$[db].
\begin{table}[h]
\caption{Rephasing -- \textbf{theoretical}/\bl{\textbf{simulated}} values for $c_2$, $c_1$, $\hat{p}_{err}^{(clup)}$, and $r$ ($n=2000$, $ 1/\sigma^2=9$[db])} \vspace{.1in}
\hspace{-0in}\centering
\footnotesize{
\begin{tabular}{||c||c||c|c||c|c||c|c||c|c||}\hline\hline
$ 1/\sigma^2 $[db] & $\hat{\gamma}_1\sqrt{n}  $ & $c_2$ & $c_2$ & $c_1$ & $c_1$ & $\hat{p}_{err}^{(clup)}$ & $\hat{p}_{err}^{(clup)}$ & $\frac{r}{\sqrt{n}}$ & $\frac{r}{\sqrt{n}}$\\ \hline\hline
$9 $ (phase 0) & $\mathbf{1.9024 }$ & $\mathbf{0.8343 }$ & $\bl{\mathbf{0.8353 }}$ & $\mathbf{0.8586 }$ & $\bl{\mathbf{0.8588 }}$ & $\mathbf{3.4801e-02 }$ & $\bl{\mathbf{3.5162e-02 }}$ & $\mathbf{0.2138 }$ & $\bl{\mathbf{0.2135 }}$ \\ \hline
$9 $ (phase 1) & $\mathbf{0.7916 }$ & $\mathbf{0.9432 }$ & $\bl{\mathbf{0.9435 }}$ & $\mathbf{0.9437 }$ & $\bl{\mathbf{0.9398 }}$ & $\mathbf{1.7933e-02 }$ & $\bl{\mathbf{2.0148e-02 }}$ & $\mathbf{0.2624 }$ & $\bl{\mathbf{0.2620 }}$ \\ \hline\hline
\end{tabular}}
\label{tab:tabclup4}
\end{table}
\begin{table}[h]
\caption{Rephasing -- \textbf{theoretical}/\bl{\textbf{simulated}} values for $c_2$, $c_1$, $\hat{p}_{err}^{(clup)}$, and $r$ ($n=2000$, $ 1/\sigma^2=10$[db])} \vspace{.1in}
\hspace{-0in}\centering
\footnotesize{
\begin{tabular}{||c||c||c|c||c|c||c|c||c|c||}\hline\hline
$ 1/\sigma^2 $[db] & $\hat{\gamma}_1\sqrt{n}  $ & $c_2$ & $c_2$ & $c_1$ & $c_1$ & $\hat{p}_{err}^{(clup)}$ & $\hat{p}_{err}^{(clup)}$ & $\frac{r}{\sqrt{n}}$ & $\frac{r}{\sqrt{n}}$\\ \hline\hline
$10 $ (phase 0) & $\mathbf{1.2949 }$ & $\mathbf{0.8937 }$ & $\bl{\mathbf{0.8943 }}$ & $\mathbf{0.9218 }$ & $\bl{\mathbf{0.9223 }}$ & $\mathbf{1.0500e-02 }$ & $\bl{\mathbf{1.0354e-02 }}$ & $\mathbf{0.2079 }$ & $\bl{\mathbf{0.2079 }}$ \\ \hline
$10 $ (phase 1) & $\mathbf{0.4657 }$ & $\mathbf{0.9910 }$ & $\bl{\mathbf{0.9912 }}$ & $\mathbf{0.9898 }$ & $\bl{\mathbf{0.9896 }}$ & $\mathbf{3.6882e-03 }$ & $\bl{\mathbf{3.8263e-03 }}$ & $\mathbf{0.2702 }$ & $\bl{\mathbf{0.2701 }}$  \\ \hline\hline
\end{tabular}}
\label{tab:tabclup5}
\end{table}

\subsection{Increasing the problem dimensions}
\label{sec:incdim}

While it is clear from the above discussion about the $\text{CLuP}^{r_0}$ that it is well-suited for large scale applications, we below show what effect the increasing of problem dimensions has on the overall performance. In Table \ref{tab:tabclup6} and Figure \ref{fig:LSclupprerrfunn1} we show the change in probability of error as $n$ increases for $ 1/\sigma^2=9$[db]. On the other hand, the same type of change for $1/\sigma^2=10$[db] we show in Table \ref{tab:tabclup7} and Figure \ref{fig:LSclupprerrfunn2}.
\begin{table}[h]
\caption{Increasing the problem dimension -- \textbf{theoretical}/\bl{\textbf{simulated}} values for $\hat{p}_{err}^{(clup)}$ ($ 1/\sigma^2=9$[db])} \vspace{.1in}
\hspace{-0in}\centering
\small{
\begin{tabular}{||c||c|c|c|c|c||c||}\hline\hline
$n$ & $300$ & $500$ & $1000$ & $2000$ & $4000$ & $\infty$ (\textbf{limit -- theory})\\ \hline\hline
$\hat{p}_{err}^{(clup)}$ (phase 0) & $\bl{\mathbf{0.0429}}$ &  $\bl{\mathbf{0.0382}}$  & $\bl{\mathbf{0.0350}}$  & $\bl{\mathbf{0.0351}}$  & $\bl{\mathbf{0.0346}}$ & $\mathbf{0.0348}$ \\ \hline\hline
$\hat{p}_{err}^{(clup)}$ (phase 1) & $\bl{\mathbf{0.0355}}$ &  $\bl{\mathbf{0.0292}}$  & $\bl{\mathbf{0.0226}}$  & $\bl{\mathbf{0.0201}}$  & $\bl{\mathbf{0.0185}}$ & $\mathbf{0.0179}$ \\ \hline\hline
\end{tabular}}
\label{tab:tabclup6}
\end{table}
\begin{figure}[htb]
\centering
\centerline{\epsfig{figure=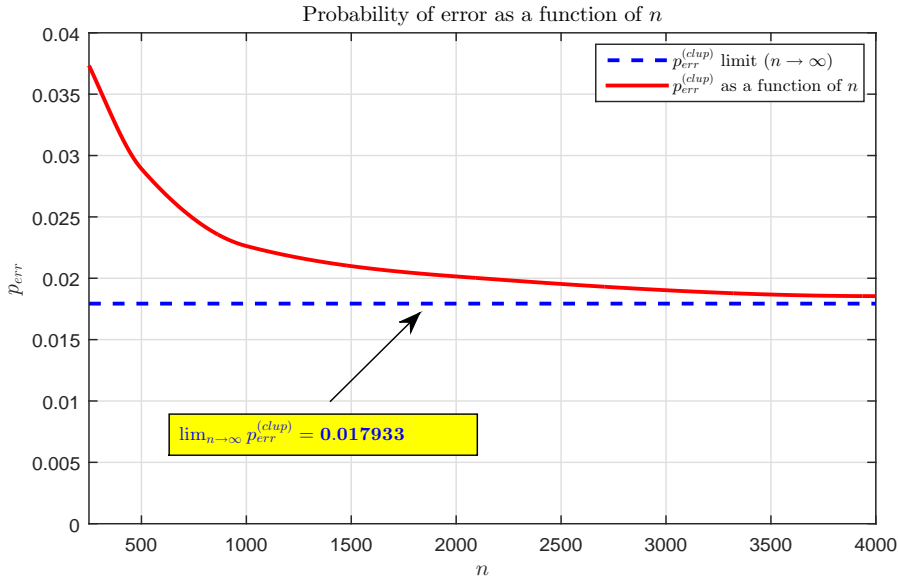,width=13.5cm,height=8cm}}
\caption{$p_{err}$ as a function of $n$; $\alpha=0.8$, $1/\sigma^2=9$[db]}
\label{fig:LSclupprerrfunn1}
\end{figure}
As can be seen from both sets of tables and figures, already for $n=500$ one achieves almost optimal performance. This means that although $\text{CLuP}^{r_0}$ is particularly tailored to handle large dimensions, it can fairly accurate work even for moderately small $n$. Moreover, as SNR increases one generically expects that smaller dimensions are sufficient to achieve the same level of error. Both tables and figures confirm that this is indeed true.
\begin{table}[h]
\caption{Increasing the problem dimension -- \textbf{theoretical}/\bl{\textbf{simulated}} values for $\hat{p}_{err}^{(clup)}$ ($ 1/\sigma^2=10$[db])} \vspace{.1in}
\hspace{-0in}\centering
\small{
\begin{tabular}{||c||c|c|c|c||c||}\hline\hline
$n$ & $400$ & $500$ & $1000$ & $2000$ & $\infty$ (\textbf{limit -- theory})\\ \hline\hline
$\hat{p}_{err}^{(clup)}$ (phase 0) & $\bl{\mathbf{0.0122}}$ &  $\bl{\mathbf{0.0116}}$  & $\bl{\mathbf{0.0106}}$  & $\bl{\mathbf{0.0104}}$  &  $\mathbf{0.0105}$ \\ \hline\hline
$\hat{p}_{err}^{(clup)}$ (phase 1) & $\bl{\mathbf{6.50e-3}}$ &  $\bl{\mathbf{5.47e-3}}$  & $\bl{\mathbf{4.44e-3}}$  & $\bl{\mathbf{3.83e-3}}$  & $\mathbf{3.69e-3}$ \\ \hline\hline
\end{tabular}}
\label{tab:tabclup7}
\end{table}
\begin{figure}[htb]
\centering
\centerline{\epsfig{figure=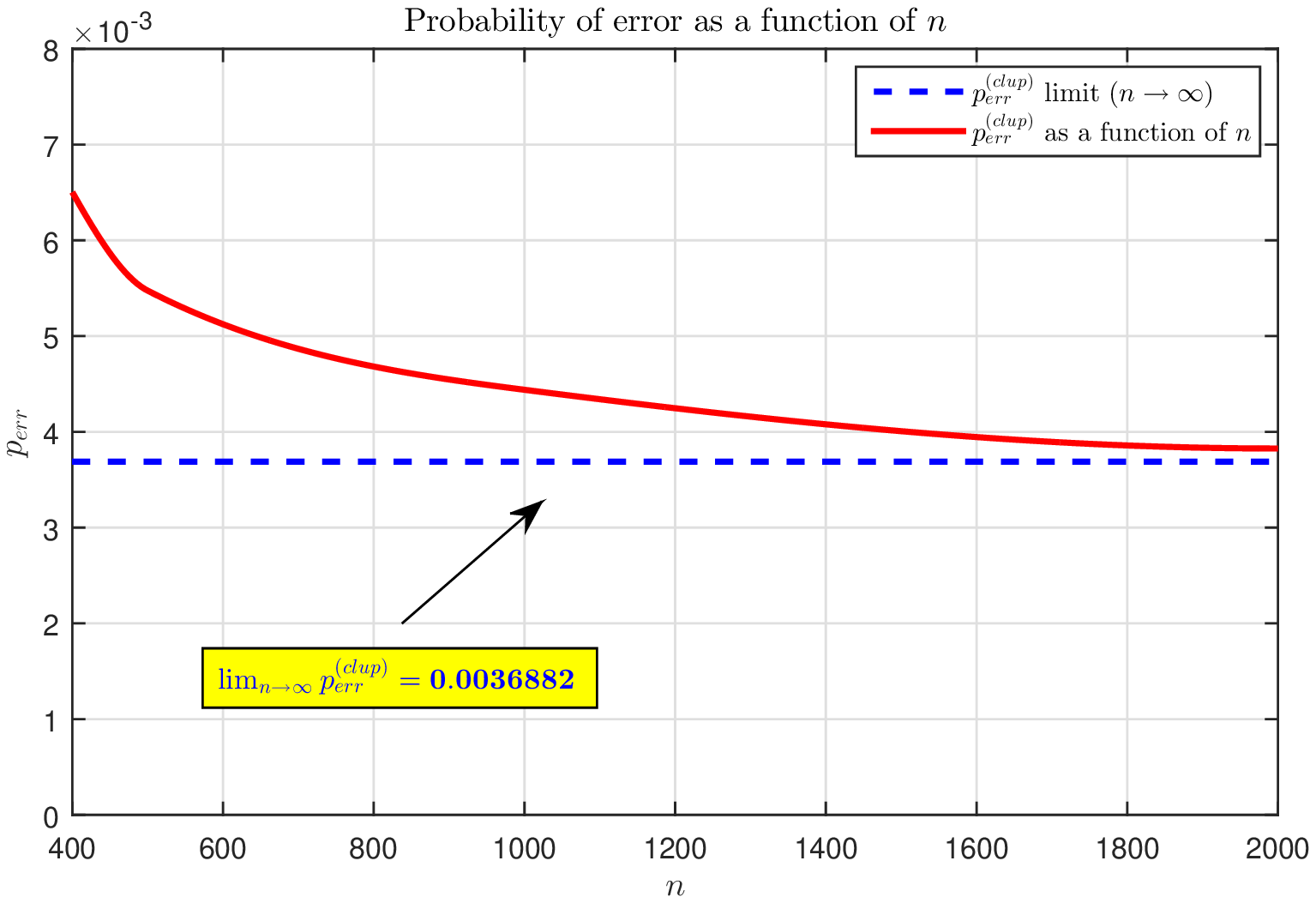,width=13.5cm,height=8cm}}
\caption{$p_{err}$ as a function of $n$; $\alpha=0.8$, $1/\sigma^2=10$[db]}
\label{fig:LSclupprerrfunn2}
\end{figure}

\section{Conclusion}
\label{sec:conc}

In this paper we revisited the CLuP mechanism that we recently introduced in \cite{Stojnicclupint19,Stojnicclupcmpl19,Stojnicclupplt19} as a fast \bl{\textbf{Random Duality Theory}} based polynomial procedure for handling hard optimization problems. As was already demonstrated in \cite{Stojnicclupint19,Stojnicclupcmpl19,Stojnicclupplt19}, although the theoretical RDT based predictions are for large dimensional problems, the CLuP performs very well even for moderate problem sizes of a few hundreds. Given our earlier success in handling large scale problems through RDT one would then expect that CLuP exhibits a similar type of behavior as well. The results that we presented here indeed confirm such expectations.

We first provided a theoretical foundation for the algorithmic use of RDT. Along the lines of our earlier works, we first reemphasized one of great RDT features, namely, its ability to transform constrained optimization problems into unconstrained ones. We recalled on some of the key technical components behind the thinking process that enabled such a transformation and then showed how these technicalities can be redone for the concrete MIMO ML detection problem. After redoing all the technicalities we established the resulting CLuP adaptation as $\text{CLuP}^{r_0}$ or the large scale CLuP.

Once the key ideas for creating the large scale CLuP are established there are many practical ways how one can go about the concrete implementation. We chose the above mentioned $\text{CLuP}^{r_0}$, as a scaled regularized contraction optimization mechanism based on a mixture of stationary points plus constraints satisfaction principles. Since this is the introductory paper regarding the large scale CLuP considerations we chose a relatively simple mechanism to showcase the main ideas in the simplest possible way. We also designed many other more sophisticated implementations and will discuss the most successful of them in some of our companion papers.

It is worth noting though, that even the simplest introductory variant that we presented here has some of generally very desirable large scale features. First, it achieves basically theoretically minimal complexity per iteration of $mn$ basic multiplication/addition operations (essentially only a single matrix/vector multiplication); second, a fairly small number of iterations (of order of a few hundreds) is often sufficient to achieve a decent level of performance.

Another particularly important and practically useful feature of the introduced $\text{CLuP}^{r_0}$ mechanism is the appearance of the so-called \emph{\textbf{rephasing}} phenomenon. It turns out that due to the algorithm's structure one can rerun the $\text{CLuP}^{r_0}$ with different set of running parameters and obtain the desired performance. In the problem instances that we studied here that wasn't generally necessary. However, there are scenarios where such a rerunning might be of crucial importance to ensure that CLuP maintains the ML achieving ability. In separate papers we will consider some of such scenarios and discuss in great details all the important rephasing aspects.

We also provided a solid set of numerical results which confirmed all theoretical predictions. Moreover, we demonstrated that even the large scale CLuP maintains ability to handle problems of smaller dimensions (of order of several hundreds). Of course, we showcased that for dimensions of several thousands it works almost exactly as the theory predicts (the difference between the theoretical and simulated values of various system parameters is often on the fourth or fifth decimal). We would specifically like to emphasize that the large scale CLuP is expected to perform even better as the dimensions grow and problems of size of several tens/hundreds of thousands/millions can be handled even with a better precision. This is of course especially important in the \emph{\textbf{big data}} era where such problem instances are of particular interest.

Now that all the key concepts are available options for further explorations are almost endless. As mentioned above, various other implementations are possible. Many of them we have already created and will discuss in separate papers. Moreover, the application of the main concepts is in no way restricted to the MIMO ML detection. We chose MIMO ML as the starting point of the large scale CLuP story to in a way follow the way how we introduced the original CLuP. However, the whole theory is applicable to a large set of problems from various other scientific fields. Quite a few of such applications we have already explored and will present the main conclusions of such explorations in separate papers. Here we would just like to emphasize that similarly to the basic CLuP, the large scale CLuP when used in different types of optimization problems requires a bit of technical modification but the core of the ideas is what we presented in this paper and in a long line of work \cite{StojnicCSetam09,StojnicCSetamBlock09,StojnicISIT2010binary,StojnicDiscPercp13,StojnicUpper10,StojnicGenLasso10,StojnicGenSocp10,StojnicPrDepSocp10,StojnicRegRndDlt10,Stojnicbinary16fin,Stojnicbinary16asym} and more recently \cite{Stojnicclupint19,Stojnicclupcmpl19,Stojnicclupplt19}.

\begin{singlespace}
\bibliographystyle{plain}
\bibliography{cluplargescRefs}

\begin{thebibliography}{10}

\bibitem{BunTsyWeg07}
F.~Bunea, A.~B. Tsybakov, and M.~H. Wegkamp.
\newblock Sparsity oracle inequalities for the lasso.
\newblock {\em Electronic Journal of Statistics}, 1:169--194, 2007.

\bibitem{CheDon95}
S.S. Chen and D.~Donoho.
\newblock Examples of basis pursuit.
\newblock {\em Proceeding of wavelet applications in signal and image
  processing III}, 1995.

\bibitem{DonMalMon10}
D.~Donoho, A.~Maleki, and A.~Montanari.
\newblock The noise-sensitiviy phase transition in compressed sensing.
\newblock available online at \bl{\url{http://arxiv.org/abs/1004.1218}}.

\bibitem{FinPhoSD85}
U.~Fincke and M.~Pohst.
\newblock Improved methods for calculating vectors of short length in a
  lattice, including a complexity analysis.
\newblock {\em Mathematics of Computation}, 44:463--471, April 1985.

\bibitem{GoeWill95}
M.~Goemans and D.~Williamnson.
\newblock Improved approximation algorithms for maximum cut and satisfiability
  problems using semidefinite programming.
\newblock {\em Journal of ACM}, 42(6):1115--1145, 1995.

\bibitem{GolVanLoan96Book}
G.~Golub and C.~Van Loan.
\newblock {\em Matrix Computations}.
\newblock John Hopkins University Press, 3rd edition, 1996.

\bibitem{HassVik05}
B.~Hassibi and H.~Vikalo.
\newblock On the sphere decoding algorithm. {P}art {I}: The expected
  complexity.
\newblock {\em IEEE Trans. on Signal Processing}, 53(8):2806--2818, August
  2005.

\bibitem{JalOtt05}
J.~Jalden and B.~Ottersten.
\newblock On the complexity of the sphere decoding in digital communications.
\newblock {\em IEEE Trans. on Signal Processing}, 53(4):1474--1484, August
  2005.

\bibitem{GroLovSch93Book}
L.~Lovasz M.~Grotschel and A.~Schriver.
\newblock {\em Geometric algorithms and combinatorial optimization}.
\newblock New York: Springer-Verlag, 2nd edition, 1993.

\bibitem{StojnicCSetamBlock09}
M.~Stojnic.
\newblock Block-length dependent thresholds in block-sparse compressed sensing.
\newblock available online at \bl{\url{http://arxiv.org/abs/0907.3679}}.

\bibitem{StojnicDiscPercp13}
M.~Stojnic.
\newblock Discrete perceptrons.
\newblock available online at \bl{\url{http://arxiv.org/abs/1306.4375}}.

\bibitem{StojnicGenLasso10}
M.~Stojnic.
\newblock A framework for perfromance characterization of \emph{LASSO}
  algortihms.
\newblock available online at \bl{\url{http://arxiv.org/abs/1303.7291}}.

\bibitem{StojnicGenSocp10}
M.~Stojnic.
\newblock A performance analysis framework for \emph{SOCP} algorithms in noisy
  compressed sensing.
\newblock available online at \bl{\url{http://arxiv.org/abs/1304.0002}}.

\bibitem{StojnicPrDepSocp10}
M.~Stojnic.
\newblock A problem dependent analysis of \emph{SOCP} algorithms in noisy
  compressed sensing.
\newblock available online at \bl{\url{http://arxiv.org/abs/1304.0480}}.

\bibitem{StojnicRegRndDlt10}
M.~Stojnic.
\newblock Regularly random duality.
\newblock available online at \bl{\url{http://arxiv.org/abs/1303.7295}}.

\bibitem{StojnicUpper10}
M.~Stojnic.
\newblock Upper-bounding $\ell_1$-optimization weak thresholds.
\newblock available online at \bl{\url{http://arxiv.org/abs/1303.7289}}.

\bibitem{StojnicCSetam09}
M.~Stojnic.
\newblock Various thresholds for $\ell_1$-optimization in compressed sensing.
\newblock available online at \bl{\url{http://arxiv.org/abs/0907.3666}}.

\bibitem{StojnicISIT2010binary}
M.~Stojnic.
\newblock Recovery thresholds for $\ell_1$ optimization in binary compressed
  sensing.
\newblock {\em ISIT, IEEE International Symposium on Information Theory}, pages
  1593 -- 1597, 13-18 June 2010.
\newblock Austin, TX.

\bibitem{Stojnicbinary16asym}
M.~Stojnic.
\newblock Box constrained $\ell_1$ optimization in random linear systems --
  asymptotics.
\newblock 2016.
\newblock available online at \bl{\url{http://arxiv.org/abs/1612.06835}}.

\bibitem{Stojnicbinary16fin}
M.~Stojnic.
\newblock Box constrained $\ell_1$ optimization in random linear systems --
  finite dimensions.
\newblock 2016.
\newblock available online at \bl{\url{http://arxiv.org/abs/1612.06839}}.

\bibitem{Stojnicclupcmpl19}
M.~Stojnic.
\newblock Complexity analysis of the controlled loosening-up ({CLuP})
  algorithm.
\newblock 2019.
\newblock available online at \bl{\url{http://arxiv.org/abs/1909.01190}}.

\bibitem{Stojnicclupint19}
M.~Stojnic.
\newblock Controlled loosening-up ({CLuP}) -- achieving \emph{exact} {MIMO ML}
  in polynomial time.
\newblock 2019.
\newblock available online at \bl{\url{http://arxiv.org/abs/1909.01175}}.

\bibitem{Stojnicclupplt19}
M.~Stojnic.
\newblock Starting {CLuP} with polytope relaxation.
\newblock 2019.
\newblock available online at \bl{\url{http://arxiv.org/abs/1909.01201}}.

\bibitem{StojnicBBSD05}
M.~Stojnic, Haris Vikalo, and Babak Hassibi.
\newblock A branch and bound approach to speed up the sphere decoder.
\newblock {\em ICASSP, IEEE International Conference on Acoustics, Signal and
  Speech Processing}, 3:429--432, March 2005.

\bibitem{StojnicBBSD08}
M.~Stojnic, Haris Vikalo, and Babak Hassibi.
\newblock Speeding up the sphere decoder with ${H}^{\infty}$ and ${SDP}$
  inspired lower bounds.
\newblock {\em IEEE Transactions on Signal Processing}, 56(2):712--726,
  February 2008.

\bibitem{Tibsh96}
R.~Tibshirani.
\newblock Regression shrinkage and selection with the lasso.
\newblock {\em J. Royal Statistic. Society}, B 58:267--288, 1996.

\bibitem{vandeGeer08}
S.~van~de Geer.
\newblock High-dimensional generalized linear models and the lasso.
\newblock {\em Ann. Statist.}, 36(2):614--645, 2008.

\bibitem{vanMaarWar00}
H.~van Maaren and J.P. Warners.
\newblock Bound and fast approximation algorithms for binary quadratic
  optimization problems with application on {MAX 2SAT}.
\newblock {\em Discrete applied mathematics}, 107:225--239, 2000.

\end{thebibliography}
\end{singlespace}

\end{document}